\RequirePackage{fix-cm}
\documentclass[smallextended]{svjour3}       

\pdfoutput=1
\usepackage[pdftex]{graphicx}

\usepackage{comment}
\usepackage{color}
\usepackage[utf8]{inputenc}
\usepackage{amsmath, amssymb}
\usepackage[]{natbib}
\bibpunct[:]{(}{)}{;}{a}{}{,}

\def \fext{pdf}
\def \cng{black}
\def \cngb{black}
\newcommand{\capup}{\vspace{-3truemm}}

\newcommand{\argmin}{\mathop{\rm arg~min}\limits}
\newcommand{\ison}{{\rm ~is~on~}}

\def\vect#1{\mbox{\boldmath $#1$}}
\def\mc#1{\mathcal{#1}}
\usepackage[title,titletoc]{appendix}

\smartqed  
\begin{document}

\title{Self-Fulfilling Signal of an Endogenous State in Network Congestion Games}



\author{Tatsuya Iwase         \and
        Yukihiro Tadokoro \and
        Daisuke Fukuda
}


\institute{T. Iwase, Y. Tadokoro \at
	Toyota Central R\&D Labs., Inc. \\
              Aichi, 480-1192 Japan \\
              \email{tiwase@mosk.tytlabs.co.jp}           
           \and
           D. Fukuda \at
           Tokyo Institute of Technology \\
           Tokyo, 152-8550 Japan
}

\date{Received: date / Accepted: date}

\maketitle

\begin{abstract}
We consider the problem of coordination via signaling in network congestion games to improve social welfare deteriorated by incomplete information about traffic flow. Traditional studies on signaling, which focus on exogenous factors of congestion and ignore congestion externalities, fail to discuss the oscillations of traffic flow. To address this gap, we formulate a problem of designing a coordination signal on endogenous information about traffic flow and introduce a {\it self-fulfilling} characteristic of a signal that guarantees an outcome flow consistent with the signal itself without causing the unwanted oscillation. An instance of the self-fulfilling signal is shown in the case of a Gaussian signal distribution. In addition, we show simple numerical examples. The results reveal how a self-fulfilling signal suppresses the oscillation and simultaneously improves social welfare through improved network efficiency.
\end{abstract}

\section{Introduction}
\label{intro}

\label{sec:intro}

%
%
%
%
Network congestion games (also called selfish routing games) model route choice behaviors of self-interested agents on a graph. One prominent application of network congestion games is in the studies on routing in transportation or telecommunication networks. These studies conclude that selfish behaviors by agents cause congestion and inefficiency in a network \citep{roughgarden}. Therefore, the coordination problem emerges among selfish agents to mitigate the inefficient use of the network and improve overall social welfare.

One way to realize such coordination is to charge tolls on the congested edges \citep{toll}. However, most previous studies are based on the assumption that agents have complete information on other agent types and can therefore predict the flow of all edges in equilibrium. In contrast, agents sometimes make incorrect decisions in real situations because they do not have information regarding other agents in the network and regarding which routes are congested at any given moment. In such a case, incomplete information causes inefficiency \citep{angeletos}.

Signaling is a possible solution to eliminate the inefficiency arising from incomplete information by sending signals to agents to provide traffic flow information for all edges. When applied to transportation networks, car navigation systems or smartphones can be used to implement signals. However, this solution is not comprehensive because difficulties also arise due to congestion externalities. Since traffic flow is an endogenous variable of the game, a signal about a flow can change the outcome of the flow itself. For example, if a car navigation system reports that one of two routes is congested, agents will then choose the other route, which in turn is likely to become congested due to higher traffic volumes. This results in congestion oscillating between routes, which is called {\it oscillation} of traffic flow or {\it hunting phenomenon} \citep{klein, fukuda, oscil}. This oscillation generates outcome flows that differ from those informed by the signal, thus turning the system into an inadvertent liar.

Designing a traffic flow signal that results in the same flow outcome as the pre-committed signal is necessary to avoid the oscillation. This signal is referred to as a {\it self-fulfilling} signal as it has the capacity to realize an outcome equal to the signal itself. Our goal is to design a self-fulfilling signal that mitigates the inefficiency caused by incomplete traffic flow information.

\subsection{Our Contributions and Paper Structure}
\label{sec:contribution}

We examine a self-fulfilling signal in a network congestion game to mitigate inefficiency caused by incomplete information. Our contribution is threefold. First, we formulate the problem of designing a self-fulfilling signal for an endogenous game variable. This problem formulation is general and includes many applications that have endogenous variables with incomplete information. Second, we show an instance of an asymmetric Gaussian signal that has self-fulfilling capability in multicommodity, atomic, and unweighted network congestion games with affine cost functions and Gaussian priors. This result provides a new alternative of endogenous information to the signaling scheme. Third, we employ examples of congestion games to verify the performance of a self-fulfilling signal in mitigating inefficiency caused by incomplete information. The results indicate the possibilities of coordination by signaling with endogenous information.

Section~\ref{sec:model} presents our model. Section~\ref{sec:sf} illustrates our main results regarding self-fulfilling signals in network congestion games. Section~\ref{sec:exp} describes numerical examples and their results.

\subsection{Related Works}
\label{sec:related}

The idea of Wardrop's User Equilibrium (UE) and System Optimum (SO) \citep{wardrop} has been studied in the field of transportation, telecommunication networks and game theory \citep{altman}. Since Rosenthal first introduced congestion games \citep{rosenthal}, many varieties of derivative games have been proposed, and significant results have been obtained regarding inefficiencies in equilibrium (see e.g., \citep{roughgarden, atomic, mixed}) as well as regarding coordination of agents to optimize social welfare (see e.g., \citep{so}). However, most such studies assume complete information about traffic flow to calculate equilibria. A review of previous related studies reveals that the most popular coordination mechanism is toll collection \citep{toll,dynaprice}. However, this approach introduces costs associated with establishing the collection infrastructure as well as requires charging responsibilities.

Coordination between agents and emergence of consensus are also topics that have been extensively researched. Consensus algorithms \citep{consensus}, El Farol bar problems \citep{elfarol}, minority games \citep{minor}, and multiagent Markov decision process \citep{mamdp} are all examples of coordination problems among autonomous agents without any coordinator. These studies assume the involvement of a constant set of agents and that these agents iteratively play the same game, updating their strategies to avoid conflicts with other agents through reinforcement learning.
\color{\cng}
However, scenarios also exist wherein agents fail to learn coordinated strategies. Former studies showed that reinforcement learning becomes computationally intractable and algorithm may not converge when agents do not have full observation of whole states, such as traffic of all edges \citep{pomdp,pomdp2,pomdp3}. In an extreme case wherein agents experience the game only once, such as when they drive to popular leisure venue on vacation, they do not have the opportunity to learn road traffic iteratively.
\color{black}

Coordination problems with a mediator have also been examined in leader-follower games \citep{lfmdp} and Stackelberg games \citep{isg}. In these studies, a leader makes a decision that affects followers’ rewards such that followers are motivated to coordinate. However, most of these studies assume complete information on the state of the environment and that both the leader and followers can observe the state. This is not always true in real road networks.

Recently, a new coordinating approach has been proposed where a mediator utilizes private information as a signal to persuade agents and realize coordination \citep{Kamenica}. Since this new signaling approach does not require infrastructure such as toll collections and thus is free from the associated initial costs, it has the potential to be quickly implemented in a range of applications, 
\color{\cng}
especially on the basis of widespread smartphone technologies or upcoming intelligent vehicles connected to the Internet. 
\color{black}
For example, \citep{kremer} demonstrated an application of signaling to road networks. A sender collects traffic flow information for each road from individual drivers and subsequently sends all drivers collated information that pertains to their route choices. However, these studies assume that traffic flow is an exogenous variable and neglect congestion externalities that cause the oscillation of traffic flow. Bayes correlated equilibrium \citep{bce} is a signaling technique that is able to manage externalities. However, it assumes that a signal is exogenous information and not endogenous information such as traffic flow.

The oscillations of the game outcome have been examined in the context of stability analysis of repeated games with the theory of Nash dynamics \citep{nd} or evolutionary game theory 
\color{\cng}
\citep{evolution, zhang}. 
\color{black}
However, these studies are based on games among autonomous agents without any coordinators. 
\color{\cng}
\citep{bayeslearn} concluded that drivers' observation errors accumulate through iterations in selfish routing under incomplete information and the traffic flow disperses over a network and becomes stable autonomously. However, this result seems optimistic when compared with reality. 
\color{black}
\citep{klein} employs noisy signaling to manage the oscillation. However, their approach is heuristic and the performance is not guaranteed by a theory. 
\citep{arg} formulated the oscillation problem as the anticipatory route guidance problem. They use variable message signs (VMS) to send drivers full traffic information. The information is limited only to drivers who pass through the nodes with VMS. This information control splits traffic into several paths and suppresses the oscillation. However, the information control is not flexible compared with the signaling approach that allows partial information disclosure with arbitrary signal distributions. 
\citep{sf} proposed a deterministic version of our self-fulfilling signal approach. However, such a deterministic signal sometimes causes unstable traffic and fails to suppress oscillation as shown in the later section of this paper. There are other studies on advanced traveler information systems (ATIS) that suppress the oscillation problem \citep{akiva, paz}. However, their simulation-based approach are not analytical and do not guarantee the convergence of oscillations. 
Studies on traffic engineering reveal that the oscillation can actually be observed in real road networks \citep{fukuda,kanamori}. Their approach against the oscillation is also heuristic and evaluated by a simulation in a specific case.

Despite the inroads made by the studies above, the oscillation remains an open problem. We therefore propose a novel theoretical approach for managing this open real-world issue in this study.
\color{\cng}
We consider situations such that agents do not have full observability, and then have biased beliefs on traffic as a result of imprecise adaptive learning.
\color{black}

\section{The Model}
\label{sec:model}

\subsection{Network Congestion Game}
\label{sec:ncg}

This section introduces the atomic, singleton, and unweighted network congestion game examined in this study. A network congestion game is defined by a tuple $\Gamma=(\mc{N},\mc{G},\mc{C},\mc{K})$, where

\begin{itemize}
	\item $\mc{N}$ denotes a finite set of agents where $|\mc{N}|=n$ and an agent $i \in \mc{N}$.
	\item $\mc{G}=(\mc{V},\mc{E})$ is a graph where $|\mc{E}|=m$ and an edge $e \in \mc{E}$.
	\item $c_{e} \in \mc{C}: \mathbb{R}_{\geq 0} \to \mathbb{R}_{\geq 0}$ denotes a cost function of an edge $e$.
	\item $\mc{K}=\mc{V} \times \mc{V}$ denotes a finite set of origin-destination (OD) pairs. $k=(o,d) \in \mc{K}$ is a pair, and $k_{i}=(o_{i},d_{i}) \in \mc{K}$ is a pair of agent $i$. $n_{k}=|\{k_{i}:k_{i}=k\}|$ denotes the number of agents whose pair is $k$.
\end{itemize}

A game is called a {\it single-commodity} flow problem \citep{so} when $|\mc{K}| = 1$ or a {\it multi-commodity} flow problem when $|\mc{K}| \geq 2$. We assume that OD pair $k$ is connected by $M_{k} \geq 1$ paths (or routes). 
A route is denoted by $r_{k} \in  \mc{R}_{k}$
\color{\cng}
, where $\mc{R}_{k}$ denotes a finite set of routes that have an OD pair $k$. 
\color{black}
A set of all possible routes for $\mc{K}$ is denoted by $\mc{R}=\underset{k \in \mc{K}}{\cup}\mc{R}_{k}$ and $|\mc{R}|=M$. We consider a singleton routing game in which each agent is restricted to choosing a single route. Accordingly, the strategy set of agent $i$ is $\mc{A}_{i}=\mc{R}_{k_{i}}$ and the
set of strategy profiles is $\mc{A}=\underset{i \in \mc{N}}{\times}\mc{A}_{i}$.

Given a strategy profile $a \in \mc{A}$, the set of users of route $r$ is denoted by $\mc{S}_{r}(a)=\{i \in \mc{N}:r=a_{i} \in \mc{A}_{i}\}$. Similarly, the set of users of edge $e$ is denoted by $\mc{S}_{e}(a)=\{i \in \mc{N}:r=a_{i}, e \ison r\}$. In addition, we define a route choice matrix $X(a)=[x_{ir}]$ corresponding to the strategy profile, the elements of which are denoted by Kronecker delta as
\begin{equation}
x_{ir}= \left \{
\begin{array}{lc}
1 & if~a_{i}=r, \\
0 & otherwise.
\end{array}
\right.
\end{equation}

In the remainder of the study, the argument $a$ is subject to be omitted. The flow on route $r$ is calculated by

\begin{equation}
h_{r}=|\mc{S}_{r}|=\sum_{i \in \mc{N}}x_{ir}.
\label{eq:rflow}
\end{equation}

The flow for all routes is denoted by $\vect{h} = (h_{1},...,h_{M})^T \in \mc{H}=\mathbb{R}_{\geq 0}^M$. Similarly, the flow on edge $e$ is $f_{e}=|\mc{S}_{e}|$ and $\vect{f} = (f_{1},...,f_{m})^T \in \mc{F}=\mathbb{R}_{\geq 0}^m$. The occurrence of edges on a route is specified by the edge-route incidence matrix $\vect{D}=[\delta_{er}]$ such that

\begin{equation}
\delta_{er}= \left \{
\begin{array}{lc}
1 & if~e \ison r, \\
0 & otherwise.
\end{array}
\right.
\end{equation}

Then the relation of flow between edges and routes is denoted by

\begin{equation}
\vect{f}=\vect{D}\vect{h}.
\label{eq:edgef}
\end{equation}

Let the cost of all edges be $\vect{c}=(c_{1}(\vect{f}),...,c_{m}(\vect{f}))$, the cost of routes $\vect{\phi}=(\phi_{1},...,\phi_{M})$ can thus be calculated as

\begin{equation}
\vect{\phi}=\vect{c}\vect{D}.
\label{eq:routec}
\end{equation}

For simplicity, we use a linear Gaussian model for Bayesian inference in the rest of the paper (see Appendix \ref{sec:gaussbayes}) and assume an affine cost function $\vect{c}(\vect{f})=\vect{\Lambda}\vect{f}+\vect{b}$ throughout.

\subsection{Incomplete Information and Signaling}
\label{sec:signal}

In network congestion games, agents select routes that minimize their expected costs. However, we examine a problem wherein agents do not have complete information regarding the types of other agents and cannot predict equilibrium traffic flow. In such a case, agents try to minimize their expected costs based on their subjective probability distributions of traffic flow. We refer to this subjective probability distribution as {\it belief}. For simplicity, we use a Gaussian distribution throughout. The prior belief of agent $i$ is denoted by
\begin{equation}
p_{0i}(\vect{f}) = N(\vect{f};\vect{\mu}_{0i},\vect{\Sigma}_{0})
\label{eq:prior}
\end{equation}
where $N(\vect{f};\vect{\mu},\vect{\Sigma})$ is the 
\color{\cng}
multivariate Gaussian probability density function of $\vect{f}$
with mean vector $\vect{\mu}$ and variance-covariance matrix $\vect{\Sigma}$. $\vect{\Sigma}_{0}$ is a common parameter to all agents. The parameter $\vect{\mu}_{0i} \in \mc{F}$ differs for each agent, and 
is independent and identically distributed (i.i.d) by being sampled from the following distribution
\color{black}
\begin{equation}
p(\vect{\mu}_{0i})=N(\vect{\mu}_{0i};\vect{\mu}_{h},\vect{\Sigma}_{h})
\label{eq:pu0i}
\end{equation}
\color{\cng}
where $\vect{\mu}_{h} \in \mc{F}$ and $\vect{\Sigma}_{h}$ are parameters common to all agents.
\color{black}

Since agents do not possess accurate information on traffic flow, their beliefs are biased and thus cause inefficiency in the network usage. This study examines a problem of a coordinator who sends asymmetric signals about traffic flow to agents to improve traffic efficiency. Following \citep{Kamenica}, we call the coordinator who sends the signal a {\it sender} and use the term {\it receivers} for agents who receive these signals. For coordination within the network, the sender commits to send an asymmetric signal $\vect{s} \in \mc{F}$ that follows the conditional distribution
\begin{equation}
\pi(\vect{s}|\vect{f})=N(\vect{s};\vect{f},\vect{\Sigma}_{s}).
\label{eq:pisf}
\end{equation}
\color{\cng}
In a transportation problem, this means that the sender announces traffic information of all edges. It may be shown, for example, in a smartphone road map.
\color{black}

However, the sender cannot observe traffic flow $\vect{f}$ because it is an outcome of the game (endogenous variable). Therefore, the sender must determine the signals to be sent following the distribution independent of the traffic flow

\begin{equation}
p(\vect{s})=N(\vect{s};\vect{\mu}_{s},\vect{\Sigma}_{s})
\label{eq:pis}
\end{equation}
\color{\cng}
where $\vect{\mu}_{s} \in \mc{F}$ and $\vect{\Sigma}_{s}$ are parameters to be determined by the sender. 
\color{black}

After receiving the signal, each receiver updates his belief according to the sender's commitment and Bayes' law as follows:
\begin{equation}
p_{i}(\vect{f}|\vect{s})=\frac{\pi(\vect{s}|\vect{f})p_{0i}(\vect{f})}{\int \pi(\vect{s}|\vect{f})p_{0i}(\vect{f})\mathrm{d}\vect{f}}.
\label{eq:post}
\end{equation}

A receiver then selects a route that minimizes expected costs as follows:

\color{\cng}
\begin{equation}
a_{i} = \argmin_{r \in A_{i}} \mathbb{E}[\phi_{r}(\vect{f})|\vect{s}],
\label{eq:rprob}
\end{equation}
where
\begin{equation}
\mathbb{E}[\phi_{r}(\vect{f})|\vect{s}]=\int p_{i}(\vect{f}|\vect{s})\phi_{r}(\vect{f})\mathrm{d}\vect{f}.
\label{eq:ephi}
\end{equation}
\color{black}

Since the expected cost $\mathbb{E}[\phi_{r}(\vect{f})|\vect{s}]$ depends upon the posterior belief in (\ref{eq:post}), the sender can control route choices of receivers and the outcome flow $\hat{\vect{f}} \in \mc{F}$.

However, if the signal is not conditional on the outcome flow $\hat{\vect{f}}$ as committed in (\ref{eq:pisf}), the sender becomes an inadvertent liar. Aggravating the situation, this causes a discrepancy between signal and outcome flow, resulting in the emergence of the oscillations.

\subsection{Our Problem}
\label{sec:problem}

Our problem in this study is therefore to design a signal distribution $p(\vect{s})$ that realizes an outcome consistent with commitment (\ref{eq:pisf}). This is formulated in general as follows.

\begin{definition}
	(Self-fulfilling condition) A signal $\vect{s}$ is self-fulfilling if and only if a commitment of the signal distribution is conditional on the outcome of an endogenous variable $\hat{\vect{f}}$:
	\begin{equation}
	\vect{s} \sim \pi(\vect{s}|\hat{\vect{f}}).
	\label{eq:sf}
	\end{equation}
	\label{def:sf}
\end{definition}

\color{\cng}
If a signal $\vect{s}$ satisfies this condition, receivers can trust the commitment $\pi(\vect{s}|\vect{f})$ and can calculate (\ref{eq:post}). Since the signal $\vect{s}$ cannot cause an increase of the expected cost $\mathbb{E}[\phi_{r}(\vect{f})|\vect{s}]$, receivers make decisions on the basis of (\ref{eq:rprob}), which can be controlled by the signal.
\color{black}
Thus, our goal is to design a signal distribution $p(\vect{s})$ that meets (\ref{eq:sf}).

\subsection{Iteration of Games}
\label{sec:iteration}

Thus far, we have considered only one-shot games. In some later sections, we will examine the case of repeated games wherein the signal is determined according to the outcome of the previous game to meet the self-fulfilling condition. In each iteration $t \in [1,...,\tau]$, the strategy profile, route choice matrix, flow of edges and routes, and cost of edges and routes are denoted by $a^t$, $X^t$, $\vect{f}^{t}$, $\vect{h}^t$, $\vect{c}^t$ and $\vect{\phi}^t$, respectively. Through each iteration, the game $\Gamma$ remains constant.

Since vehicles in a road network constantly change, we assume short-lived and myopic receivers who exist only in one game and are entirely replaced in the next game. This assumption is denoted by constant priors (\ref{eq:prior}) and a constant knowledge of signal accuracy (\ref{eq:pisf}), which means that receivers 
\color{\cng}
already have fixed beliefs and 
\color{black}
do not update them in iterations. Though this is a radical assumption, 
yet it is reasonable to model the traffic on weekends, so that the people visiting the places for the first time do not have any chance to update their knowledge. The assumption also models an aspect of weekday-traffic, wherein drivers do not have full observability of the entire traffic in the whole network. The assumption means that drivers have already learnt from everyday traffic experience before the signaling iterations. Owing to the intractability of reinforcement learning under partial observability  \citep{pomdp,pomdp2,pomdp3}, their strategies have been trapped into a local minimum, which is different from Wardrop equilibria, and their knowledge regarding traffic remains incomplete. Though learning dynamics amid local observability can converge to Wardrop equilibria in the case of linear cost functions \citep{dominique}, yet it does not hold in general. If the learning dynamics does not converge to Wardrop equilibria, then each driver can only have vague and biased knowledge as a constant prior in consequence of day-to-day traffic experiences.

Once a signaling system starts functioning, the sender determines signals based on the constant priors throughout the iterations. The drivers know that the outcome traffic is controlled by the self-fulfilling signals as committed in advance. Hence, they cannot learn about the “natural” traffic without any influence of the signals. Accordingly, the drivers cannot update their priors, which are unconditional on the signals, through the iterations.

In repeated games, the sender determines a signal $\vect{s}^t$ based on the constant priors and 
the observation of the previous outcome flow $\hat{\vect{f}}^{t-1}$. The outcome of a current game $\hat{\vect{f}}^{t}$ is determined by the current signal $\vect{s}^{t}$.

\section{Self-Fulfilling Signal}
\label{sec:sf}

The method to design a self-fulfilling signal depends on whether the sender has knowledge of receivers’ priors. In this section, we sequentially describe both cases.

\subsection{One-Shot Game: Known Prior Case}
\label{sec:kp}

In the case where the sender knows receivers’ priors, the self-fulfilling signal can be determined according to these priors. Given the game $\Gamma$, it is possible for the sender to predict the receivers’ choices and the outcome flow $\hat{\vect{f}}$ corresponding to the sender’s signal $p(\vect{s})$. Hence, the sender can determine a signal such that the outcome $\hat{\vect{f}}$ becomes spontaneously consistent with the signal.

The following lemma indicates how to predict the outcome flow corresponding to a given signal.

\begin{lemma}
	The outcome flow of single-commodity with affine cost functions is given by \\
	\begin{equation}
	\hat{\vect{f}}=n\vect{D}[cdf(\vect{0};\vect{\mu}_{r},\vect{\Sigma}_{r})]^T,
	\label{eq:fsc}
	\end{equation}
	where $cdf(\vect{y};\vect{\mu},\vect{\Sigma})$ is the multivariate normal cumulative distribution function and $\vect{\mu}_{r}$ and $\vect{\Sigma}_{r}$ are defined in (\ref{eq:paramr}).
	\label{thm:fsc}
\end{lemma}

\begin{proof}
	When a receiver receives a signal $\vect{s}$, the posterior belief of the receiver is calculated from (\ref{eq:prior}) and (\ref{eq:pisf}) by Bayes' law (\ref{eq:post}). The summary of Bayesian inference for linear Gaussian model is described in Appendix \ref{sec:gaussbayes}. The posterior belief is denoted by
	\begin{equation}
	p_{i}(\vect{f}|\vect{s})=N(\vect{f};\vect{\mu}_{is},\vect{\Sigma}),
	\end{equation}
	where
	\begin{equation}
	\left.
	\begin{array}{l}
	\vect{\mu}_{is}=\vect{\Sigma}(\vect{\Sigma}_{s}^{-1}\vect{s}+\vect{\Sigma}_{0}^{-1}\vect{\mu}_{0i}) \\
	\vect{\Sigma}=(\vect{\Sigma}_{s}^{-1}+\vect{\Sigma}_{0}^{-1})^{-1}.
	\end{array}
	\right.
	\end{equation}
	The signal $\vect{s}$ follows the distribution given in (\ref{eq:pis}). Hence, $\vect{s}$ in the posterior belief above is marginalized out as
	\begin{equation}
	p_{i}(\vect{f})=\int p_{i}(\vect{f}|\vect{s})p(\vect{s}) \mathrm{d}\vect{s}=N(\vect{f};\vect{\mu}_{i},\vect{\Sigma}_{1})
	\end{equation}
	where
	\begin{equation}
	\left.
	\begin{array}{l}
	\vect{\mu}_{i}=\vect{\Sigma}\vect{\Sigma}_{s}^{-1}\vect{\mu}_{s}+\vect{\Sigma}\vect{\Sigma}_{0}^{-1}\vect{\mu}_{0i} \\
	\vect{\Sigma}_{1}=\vect{\Sigma}+\vect{\Sigma}(\vect{\Sigma}\vect{\Sigma}_{s}^{-1})^T.
	\end{array}
	\right.
	\end{equation}
	The mean parameter of the prior $\vect{\mu}_{0i}$ follows the distribution in (\ref{eq:pu0i}). Hence, $\vect{\mu}_{0i}$ in the above belief is also marginalized out as
	\begin{equation}
	p(\vect{f})=\int p_{i}(\vect{f})p(\vect{\mu}_{0i})\mathrm{d}\vect{\mu}_{0i}=N(\vect{f};\overline{\vect{\mu}},\overline{\vect{\Sigma}}),
	\label{eq:px}
	\end{equation}
	where
	\begin{equation}
	\left.
	\begin{array}{l}
	\overline{\vect{\mu}}=\vect{\Sigma}\vect{\Sigma}_{s}^{-1}\vect{\mu}_{s}+\vect{\Sigma}\vect{\Sigma}_{0}^{-1}\vect{\mu}_{h} \\
	\overline{\vect{\Sigma}}=\vect{\Sigma}_{1}+(\vect{\Sigma}\vect{\Sigma}_{0}^{-1})\vect{\Sigma}_{h}(\vect{\Sigma}\vect{\Sigma}_{0}^{-1})^T.
	\end{array}
	\right.
	\end{equation}	
	With an affine cost function of $\vect{c}(\vect{f})=\vect{\Lambda}\vect{f}+\vect{b}$, the distribution of the edge cost is then
	\begin{equation}
	p(\vect{c})=N(\vect{c};\vect{\Lambda}\overline{\vect{\mu}}+\vect{b},\vect{\Lambda}\overline{\vect{\Sigma}}\vect{\Lambda}^T).
	\end{equation}
	From (\ref{eq:routec}), the distribution of the route cost becomes 
	\begin{equation}
	p(\vect{\phi})=N(\vect{\phi};\vect{\mu}_{\phi},\vect{\Sigma}_{\phi}),
	\end{equation}
	where
	\begin{equation}
	\left.
	\begin{array}{l}
	\vect{\mu}_{\phi}=(\vect{\Lambda}\overline{\vect{\mu}}+\vect{b})\vect{D} \\
	\vect{\Sigma}_{\phi}=\vect{D}^T(\vect{\Lambda}\overline{\vect{\Sigma}}\vect{\Lambda}^T)\vect{D}.
	\end{array}
	\right.
	\end{equation}	
	The receivers' choices are based on this cost distribution. Let $\mu_{\phi,r}$ be the element of $\vect{\mu}_{\phi}$ for route $r$. The receiver's problem (\ref{eq:rprob}) is then rewritten as
	\begin{equation}
	a_{i} = \argmin_{r \in A_{0}} \mu_{\phi,r}.
	\label{eq:rprob2}
	\end{equation}
	Since this is a single-commodity game, the strategy set of all receivers is identical, denoted by $A_{0}$. The probability of a receiver to selecting route $r$ is $p_{r}=P(a_{i}=r)=P(\mu_{\phi,r} \leq \mu_{\phi,{-r}})$ where $-r=\{r':r' \in A_{0}, r' \neq r\}$ denotes all routes except for $r$. Here, let $\vect{B}_{r}=[\beta_{j,l}]$ denote a matrix for cost comparison between route $r$ and other routes such that
	\begin{equation}
	\beta_{j,l}= \left \{
	\begin{array}{lc}
	1 & if~l = r, \\
	-1 & if~l=-r[j] \\
	0 & otherwise
	\end{array}
	\right.
	\end{equation}
	where $j \in [|A_{0}|-1], l \in [|A_{0}|]$.
	Then, $p_{r}$ is denoted with this matrix as  $p_{r}=P(\vect{B}_{r}\vect{\phi} \leq 0)$. The distribution of differential cost is
	\begin{equation}
	p(\vect{B}_{r}\vect{\phi})=N(\vect{B}_{r}\vect{\phi};\vect{\mu}_{r},\vect{\Sigma}_{r}),
	\end{equation}
	where
	\begin{equation}
	\left.
	\begin{array}{l}
	\vect{\mu}_{r}=\vect{B}_{r}\vect{\mu}_{\phi} \\
	\vect{\Sigma}_{r}=\vect{B}_{r}\vect{\Sigma}_{\phi}\vect{B}_{r}^T.
	\end{array}
	\right.
	\label{eq:paramr}
	\end{equation}
	Then the probability of route choice is
	\begin{equation}
	p_{r}=P(\vect{B}_{r}\phi \leq 0) = cdf(\vect{0};\vect{\mu}_{r},\vect{\Sigma}_{r}).
	\label{eq:pr}
	\end{equation}
	Meanwhile, the outcome flow of route $r$ is calculated from (\ref{eq:rflow}). If the number of receivers $n$ is sufficiently large, the law of large numbers gives
	\begin{equation}
	\left.
	\begin{array}{l}
	\hat{h}_{r}=\sum_{i \in N}x_{ir}=n\mathbb{E}[x_{ir}]=nP\{a_{i}=r\}=np_{r} \\
	=n*cdf(\vect{0};\vect{\mu}_{r},\vect{\Sigma}_{r}).
	\end{array}
	\right.
	\label{eq:hrsc}
	\end{equation}
	Accordingly, the flow of all routes is $\hat{\vect{h}}=n[cdf(\vect{0};\vect{\mu}_{r},\vect{\Sigma}_{r})]$. From (\ref{eq:edgef}), (\ref{eq:fsc}) is obtained. This completes the proof of Lemma \ref{thm:fsc}.
	\qed
\end{proof}

For multi-commodity games, the following lemma gives the outcome prediction.

\begin{lemma}
	The outcome flow of multi-commodity with affine cost functions is given by \\
	\begin{equation}
	\hat{\vect{f}}=\vect{D}[cdf(\vect{0};\vect{\mu}_{rk},\vect{\Sigma}_{rk})]\vect{n},
	\label{eq:fmc}
	\end{equation}
	where $\vect{n}=(n_{1},...,n_{|\mc{K}|})^T$, and $\vect{\mu}_{rk}$ and $\vect{\Sigma}_{rk}$ are defined in (\ref{eq:paramrk}).
	\label{thm:fmc}
\end{lemma}

\begin{proof}
	Since receiver's decision depends only on his prior and the signal from the sender, the probability of route choice $p_{rk}$ for each pair $k$ is calculated independently and similarly as in (\ref{eq:pr}) such that
	\begin{equation}
	p_{rk}=P(\vect{B}_{rk}\phi \leq 0) = cdf(\vect{0};\vect{\mu}_{rk},\vect{\Sigma}_{rk}),
	\label{eq:prk}
	\end{equation}
	where
	\begin{equation}
	\left.
	\begin{array}{l}
	\vect{\mu}_{rk}=\vect{B}_{rk}\vect{\mu}_{\phi} \\
	\vect{\Sigma}_{rk}=\vect{B}_{rk}\vect{\Sigma}_{\phi}\vect{B}_{rk}^T.
	\end{array}
	\right.
	\label{eq:paramrk}
	\end{equation}
	Then, the outcome flow of all routes is calculated by summing up the probabilities
	\begin{equation}
	\hat{\vect{h}}=[p_{rk}]\vect{n}.
	\end{equation}
	From (\ref{eq:edgef}), (\ref{eq:fmc}) is obtained. This completes the proof of Lemma \ref{thm:fmc}.
	\qed
\end{proof}

Lemma \ref{thm:fmc} immediately yields the following proposition regarding the self-fulfilling condition.

\begin{proposition}
	In the case of a network congestion game with Gaussian signal distribution (\ref{eq:pis}), the following condition is equivalent to the self-fulfilling condition (\ref{eq:sf}).
	\begin{equation}
	\vect{\mu}_{s}=\vect{D}[cdf(\vect{0};\vect{\mu}_{rk},\vect{\Sigma}_{rk})]\vect{n}
	\label{eq:sfconc}
	\end{equation}
	\label{thm:sf}
\end{proposition}

\begin{proof}
	In the case of Gaussian distribution, the signal distribution conditional on the outcome is denoted by $\pi(\vect{s}|\hat{\vect{f}})=N(\vect{s};\hat{\vect{f}},\vect{\Sigma}_{s})$. The signal actually follows the distribution in (\ref{eq:pis}). Then, the self-fulfilling condition (\ref{eq:sf}) in the Gaussian case is denoted by
	\begin{equation}
	\vect{\mu}_{s} = \hat{\vect{f}}.
	\label{eq:sfgauss}
	\end{equation}
	With Lemma \ref{thm:fmc}, (\ref{eq:sfconc}) is obtained. This completes the proof of Proposition \ref{thm:sf}.
	\qed
\end{proof}

Now we can obtain a self-fulfilling Gaussian signal as a solution to equation (\ref{eq:sfconc}) in the case of known priors. The existence of a solution is guaranteed by the following proposition.

\begin{proposition}
	Regarding (\ref{eq:sfconc}) as a self-map
	\begin{equation}
	\vect{\mu}_{s}=\vect{g}(\vect{\mu}_{s})=\vect{D}[cdf(\vect{0};\vect{\mu}_{rk},\vect{\Sigma}_{rk})]\vect{n},
	\label{eq:smap}
	\end{equation}
	there is a solution (fixed point $\vect{\mu}_{s0}$) such that $\vect{\mu}_{s0}=\vect{g}(\vect{\mu}_{s0})$.
	\label{thm:fpexist}
\end{proposition}

\begin{proof}
	$\vect{g}$ is a self-map $\vect{g}:\mc{F} \rightarrow \mc{F}$. According to (\ref{eq:edgef}), $\mc{F}$ is a linear transformation of $M$-dimensional box $\mc{H}$ and then convex compact. Then, by Brouwer's fixed-point theorem, the self-map (\ref{eq:smap}) has a fixed point. This completes the proof of Proposition \ref{thm:fpexist}.
	\qed
\end{proof}

\subsection{Repeated Games: Unknown Prior Case}
\label{sec:ukp}

Without knowing receivers’ priors, the sender cannot design a self-fulfilling signal based on Proposition \ref{thm:sf}. In this case, the sender has to adjust signals to be self-fulfilling in repeated games. Again, we assume a Gaussian signal but using the observation of the last outcome,

\begin{equation}
\left.
\begin{array}{l}
p(\vect{s}^{t}|\vect{f}^{t-1})=N(\vect{s}^{t};\vect{\mu}_{s}^t,\vect{\Sigma}_{s}) \\
\vect{\mu}_{s}^t=\vect{f}^{t-1}.
\end{array}
\right.
\label{eq:pisrp}
\end{equation}

The sender's commitment (\ref{eq:pisf}) is written as

\begin{equation}
\pi(\vect{s}^{t}|\vect{f}^{t})=N(\vect{s}^{t};\vect{f}^{t},\vect{\Sigma}_{s}).
\label{eq:pisfrp}
\end{equation}

Then, our goal is to meet a Gaussian version of self-fulfilling condition (\ref{eq:sfgauss}), which can be rewritten as

\begin{equation}
\vect{\mu}_{s}^t = {\vect{f}}^{t}.
\label{eq:sfgaussrp}
\end{equation}

According to (\ref{eq:fmc}), the outcome flow of the current game is the function of the signal parameters, which is denoted by

\begin{equation}
\vect{f}^{t}=\vect{g}(\vect{\mu}_{s}^t)=\vect{D}[cdf(\vect{0};\vect{\mu}_{rk},\vect{\Sigma}_{rk})]\vect{n}.
\label{eq:fmcrp}
\end{equation}

Now, we have the following proposition about self-fulfilling signals in repeated games.

\begin{proposition}
	Let $\vect{\mu}_{s0}$ be a fixed point in a self-map $\vect{\mu}_{s}=\vect{g}(\vect{\mu}_{s})$. The following condition is sufficient for a signal to converge on $\vect{\mu}_{s0}$ and meet the self-fulfilling condition (\ref{eq:sfgaussrp}).
	\color{\cng}
	\begin{equation}
	\max |\lambda_{j}| < 1
	\label{eq:sfconcrp}
	\end{equation}
	where $\vect{J}=\frac{\partial \vect{g}}{\partial \vect{\mu}_{s}}$ is the Jacobian of $\vect{g}(\vect{\mu}_{s})$ and $\vect{\lambda}=(\lambda_{1},...,\lambda_{m})$ are eigenvalues of $\vect{J}(\vect{\mu}_{s0})$.
	\color{black}	
	\label{thm:sfrp}
\end{proposition}

\begin{proof}
	From (\ref{eq:pisrp}) and (\ref{eq:fmcrp}), the outcome flow evolves via the following process
	\begin{equation}
	\vect{f}^{t}=\vect{g}(\vect{f}^{t-1}).
	\label{eq:smaprp}
	\end{equation}
	If the function $\vect{g}$ has an asymptotic stable fixed point, the outcome flow converges on that fixed point such that $\vect{f}^{t}=\vect{f}^{t-1}$. Hence with (\ref{eq:pisrp}), the self-fulfilling condition (\ref{eq:sfgaussrp}) is obtained.
	The sufficient condition that a fixed point $\vect{\mu}_{s0}$ of the function $\vect{g}$ becomes asymptotically stable is given by (\ref{eq:sfconcrp})
	\color{\cng}
	 according to the linearized stability theory \citep{stable}. 
	\color{black}
	 This completes the proof of Proposition \ref{thm:sfrp}.
	\qed
\end{proof}

In repeated games, a signal that satisfies the condition (\ref{eq:sfgaussrp}) converges on the fixed point $\vect{g}$ and becomes self-fulfilling. Otherwise, the outcome oscillation occurs. We will confirm this with an example in a later section.

\section{Examples}
\label{sec:exp}

\subsection{Minority Game}
\label{sec:mg}

In this section, we show a simple numerical example based on a minority game. A minority game \citep{minor} is a subclass of (network) congestion games, in which receivers have only two options $\mc{E}=\{1,2\}$. The game does not have graph structures, which means that $\vect{D}$ is a $2 \times 2$ identical matrix, $\mc{R}=\mc{A}_{i}=\mc{E}$ and $|\mc{K}|=1$. A game outcome can be denoted with a scalar variable of population rate $\omega \in [0,1]$ as $f_{1}=n\omega, f_{2}=n(1-\omega)$. The cost function is $\vect{c}(\omega)=(1-2\omega, 2\omega-1)$, which imposes more costs on majorities and ensures that minorities win. This simple example models situations such as the choice between two routes. The parameters of receivers’ beliefs and the sender's signal in (\ref{eq:prior}),(\ref{eq:pu0i}) and (\ref{eq:pis}) are also reduced as scalar variables $(\mu_{0i},\sigma^2_{0}),(\mu_{h},\sigma^2_{h})$ and $(\mu_{s},\sigma^2_{s})$, respectively. The sender aims to minimize the total cost of receivers, which is $|\vect{c}(\omega)|^2=-8(\omega-\omega^2)+2$. Hence, the sender’s utility can be denoted as the minus cost $v(\omega)=\omega-\omega^2$.

In this setting, the self-fulfilling condition (\ref{eq:sfconc}) in Proposition \ref{thm:sf} is reduced as follows:
\begin{equation}
\mu_{s}=cdf(0;\mu_{\phi},\sigma_{\phi}^2)
\label{eq:sfmg}
\end{equation}
where
\begin{equation}
\left.
\begin{array}{l}
\mu_{\phi}=2\sigma^2\left(\frac{\mu_{s}\strut}{\sigma^2_{s}\strut}+\frac{\mu_{h}\strut}{\sigma^2_{0}\strut}\right) \\
\sigma_{\phi}^2=4\left(\sigma^2_{1}+\frac{\sigma^4\strut}{\sigma^4_{0}\strut}\sigma^2_{h}\right) \\
\sigma^2=\frac{\sigma^2_{0}\sigma^2_{s}\strut}{\sigma^2_{0}+\sigma^2_{s}\strut} \\
\sigma^2_{1}=\sigma^2+\frac{\sigma^4\strut}{\sigma^2_{s}\strut}.
\end{array}
\right.
\end{equation}

The sufficient condition (\ref{eq:sfconcrp}) for self-fulfillment in the repeated game in Proposition \ref{thm:sfrp} is given in closed form as
\begin{equation}
\sigma^2_{s} \geq \frac{-4\pi\sigma^4_{0}+\sqrt{16\pi^2\sigma^8_{0}+8\pi(\sigma^2_{0}+\sigma^2_{h})\sigma^4_{0}}\strut}{4\pi(\sigma^2_{0}+\sigma^2_{h})\strut}.
\label{eq:sfrpmg}
\end{equation}
Note that the signal parameter in (\ref{eq:sfrpmg}) is only variance $\sigma^2_{s}$.

\subsubsection{Known Prior Case}
\label{sec:ekp}

Here, we show a simple numerical example of a known prior case. Parameter values are shown in Table \ref{tbl:param}. The mean parameter of receivers’ prior is $\mu_{h}=0.3$, which means that most receivers believe that route $1$ is relatively vacant and therefore preferable. In the case without signaling, receivers rush to route $1$ and the outcome percentage of the population on that route becomes $\hat{\omega}=0.85$.

\begin{table}[h]
	\centering
	\caption{Values of parameters}
	\begin{tabular}{|c|c|} \hline
		Parameter & Value \\ \hline \hline
		$\mu_{h} $ & $0.3$ \\ \hline
		$\sigma_{h}$ & $0.2$ \\ \hline
		$\sigma_{0}$ & $0.2$ \\ \hline
		$\sigma_{s}$ & $0.22$ \\ \hline
		$n$ & $81$ \\ \hline
		no. of signal samples & $150$ \\ \hline
	\end{tabular}
	\label{tbl:param}
\end{table}

The sender tries to mitigate congestion on route $1$ by sending signal $s \in [0,1]$ on endogenous variable $\omega$. The signal is designed to meet the self-fulfilling condition (\ref{eq:sfmg}) to suppress the outcome oscillations, which derives the mean parameter $\mu_{s}=0.62$. The resulting outcome becomes $\hat{\omega} = 0.63$, indicating that the signal meeting the self-fulfilling condition successfully realizes the outcome that equals itself with a slight sampling error. As long as the sender uses a self-fulfilling signal, the oscillations will always be suppressed.

Figure \ref{fig:efficiency} indicates a comparison of traffic efficiency $v(\hat{\omega})$ between the case with signal and the case without a signal. Network efficiency is increased by signaling, which means that traffic inefficiencies caused by biased priors is mitigated by sender’s self-fulfilling signal.

\begin{figure}[h]
	\centering
	\includegraphics[clip,width=7cm]{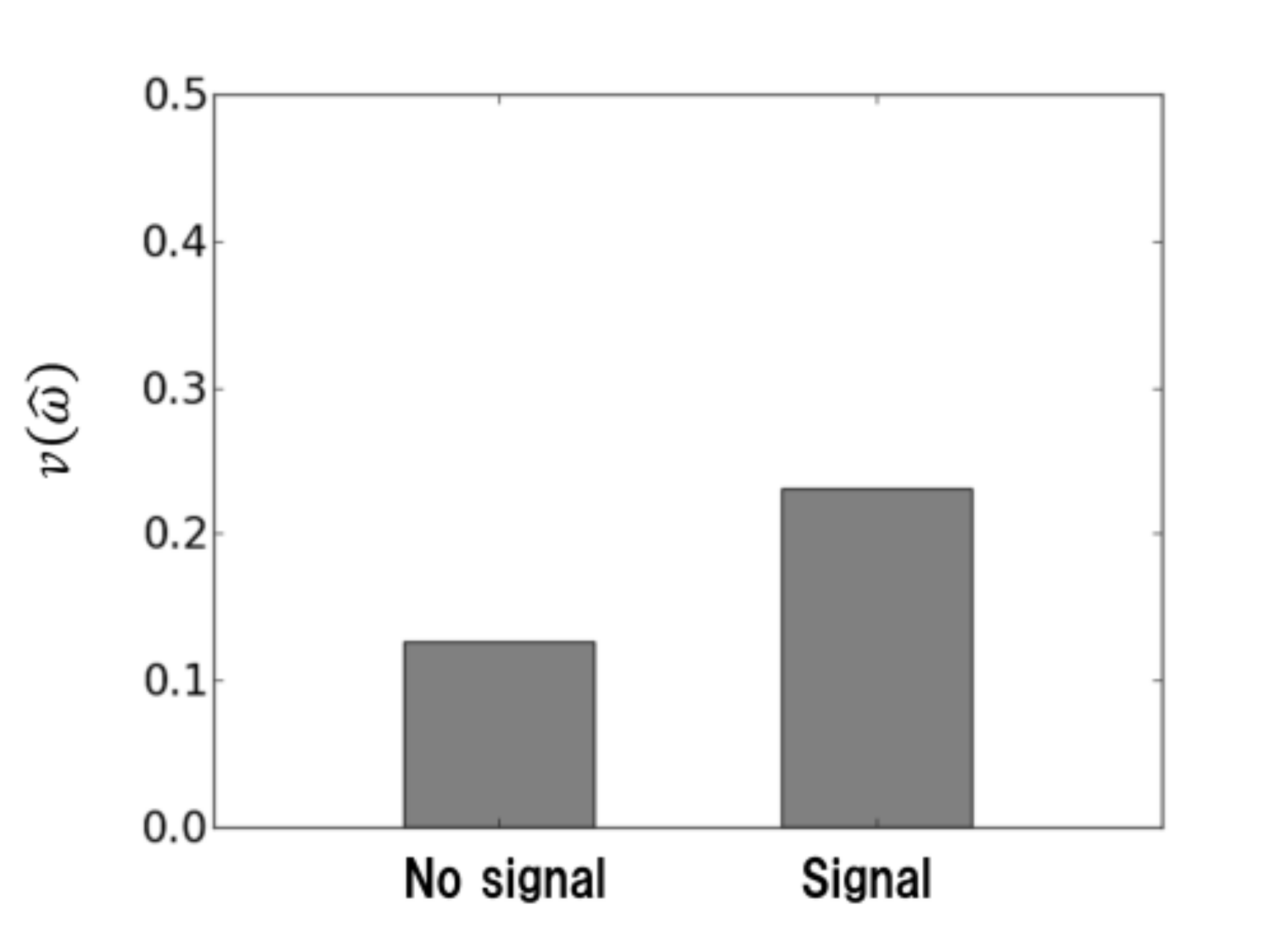}
	\capup
	\caption{Change in traffic efficiency by a self-fulfilling signal}
	\label{fig:efficiency}	
\end{figure}

\subsubsection{Unknown Prior Case}
\label{sec:eukp}

Here, we show a numerical example for an unknown prior case. We use the same parameter values as in Table \ref{tbl:param} except for signal deviation $\sigma_{s}$. Since the sender does not know receivers’ priors, the signal is determined through repeated games. We use $\mu_{s0}=0.3$ as an initial value for the signal mean.

The self-fulfilling condition (\ref{eq:sfrpmg}) is $\sigma_{s} \geq 0.2$. First, we tested $\sigma_{s}=0.13$, which does not meet the self-fulfilling condition. Figure \ref{fig:hunting} shows the evolution of outcome $\hat{\omega}$. As the theory indicates, the oscillation occurs due to the non-self-fulfilling signal.

\begin{figure}[h]
	\centering
	\includegraphics[clip,width=7cm]{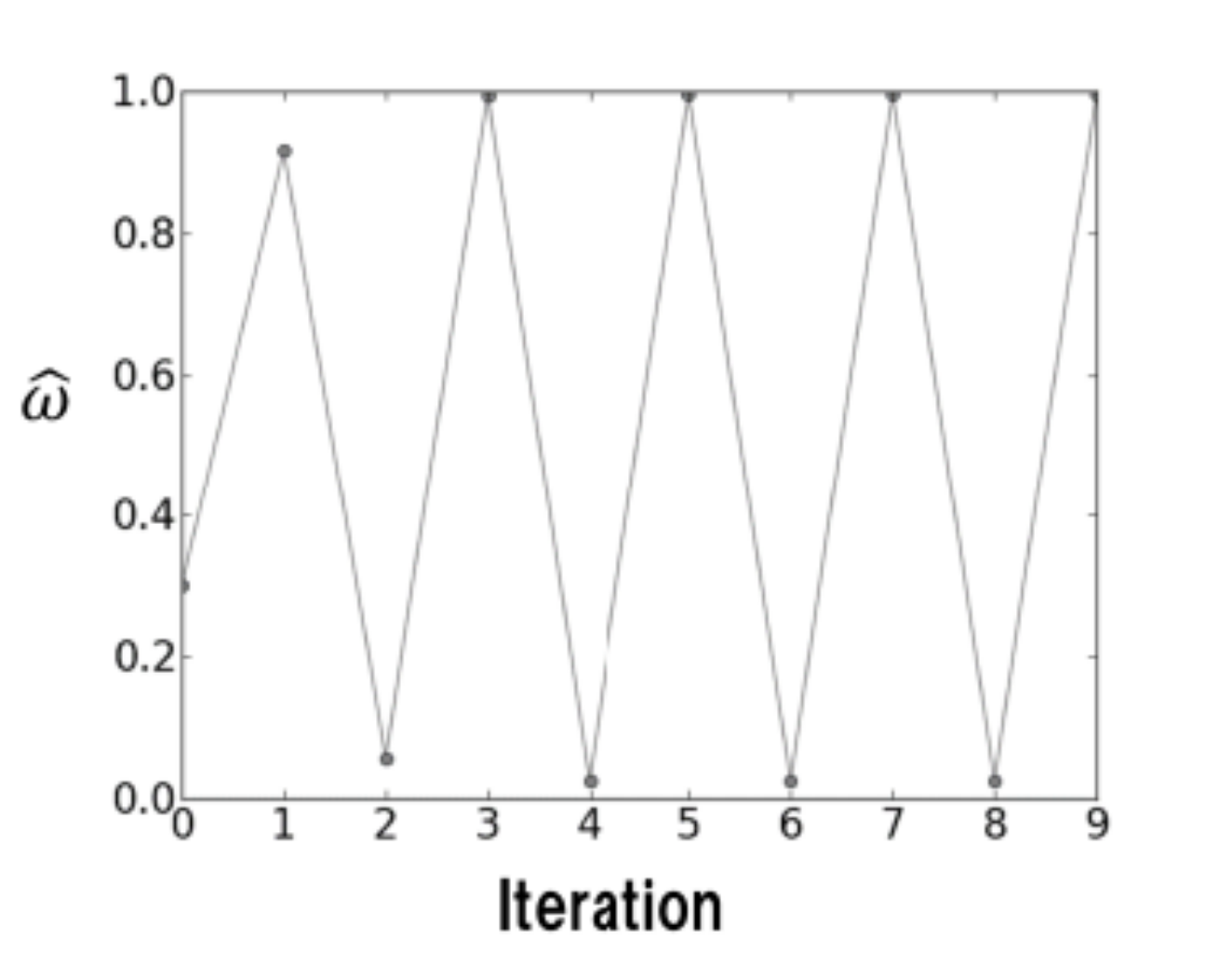}
	\capup
	\caption{Oscillation caused by non-self-fulfilling signal}
	\label{fig:hunting}	
\end{figure}

Next, we tested $\sigma_{s}=0.29$, which meets the self-fulfilling condition. Figure \ref{fig:stable} shows the result. As the theory indicates, the oscillations are suppressed and the outcome $\hat{\omega}$ converges on the same value of $\mu_{s}=0.62$, which is given by the self-fulfilling signal in the known prior case. Hence, the sender can suppress the oscillations and ensure that the outcome will converge even without knowing receivers' priors.

\begin{figure}[h]
	\centering
	\includegraphics[clip,width=7cm]{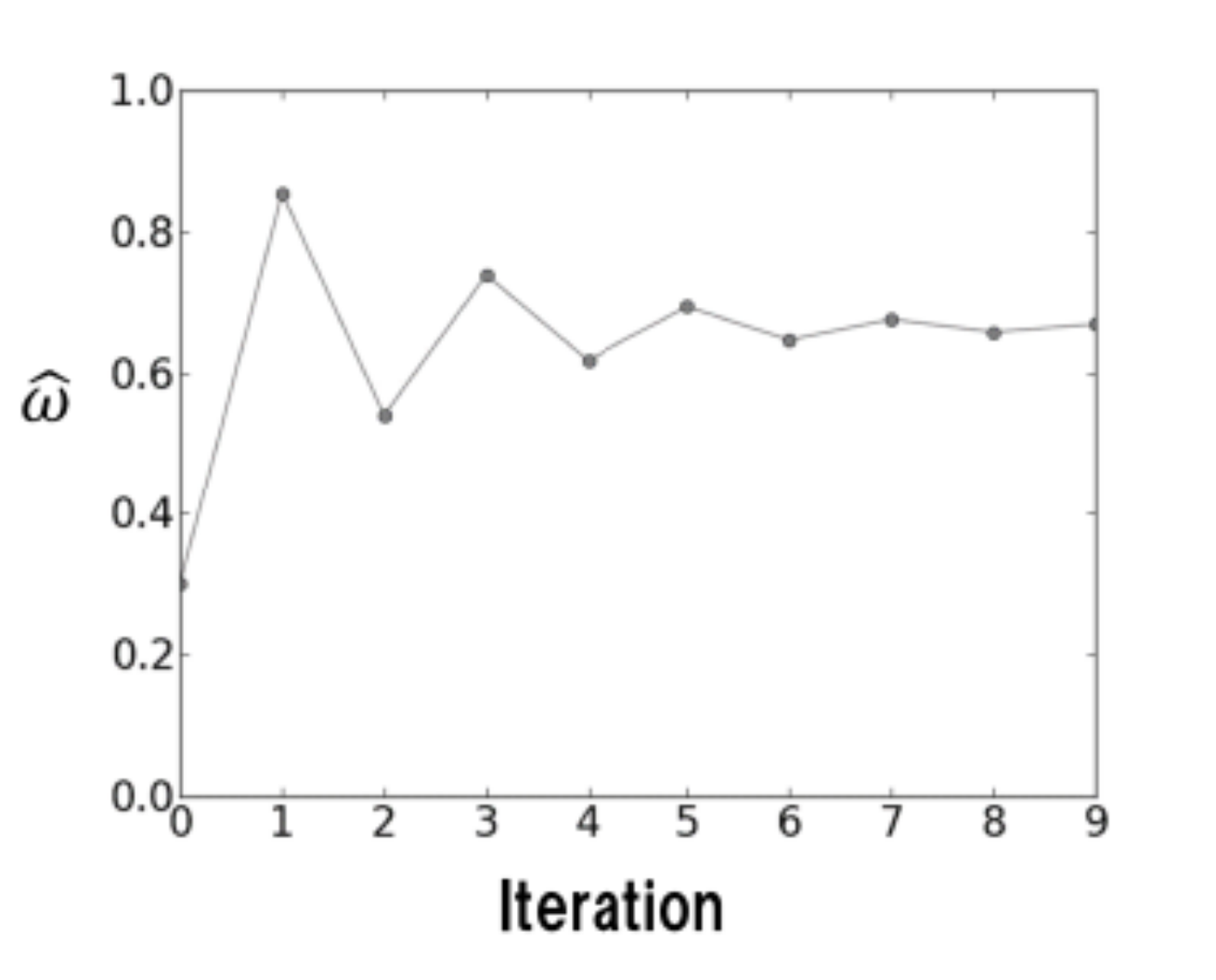}
	\capup
	\caption{Convergence of traffic led by a self-fulfilling signal}
	\label{fig:stable}	
\end{figure}

From a signaling point of view, the results are interesting because the partial disclosure of flow information as in Figure \ref{fig:stable} brings greater efficiency than in the no-disclosure case shown in Figure \ref{fig:efficiency} or in the full-disclosure case suffering from the outcome oscillations as in Figure \ref{fig:hunting}.

\color{\cng}
\subsection{Network Congestion Game}
\label{sec:encg}

In this section, we show another numerical example based on a network congestion game. Figure \ref{fig:nets} shows a bidirected graph $\mc{G}$ used in this example, which is a subset of the network used in \citep{network}
\color{\cngb}
and has $m=46$ edges.
\color{black}
There are $n=172$ players and $|\mc{K}|=14$ OD pairs 
\color{\cngb}
which are also shown in Figure \ref{fig:nets}.
\color{black}
$|\mc{R}_{k}|=5$ choices for each of the OD pairs, which are the first five unweighted shortest paths.
\color{\cngb}
We use a cost function $\vect{c}(\vect{f})=\vect{\Lambda}\vect{f}+\vect{b}$. The parameter $\vect{\Lambda}$ is a constant $m \times m$ diagonal matrix that has randomly generated elements $\Lambda_{ii} \in [2,3], i \in \{1 \dots m\}$. The parameter $\vect{b}$ is a constant vector that also has randomly generated elements $b_{i} \in [0,1]$.
\color{black}
The parameter $\vect{\Sigma}_{h}$ is generated by simulating random route choices 100 times. We let $\vect{\Sigma}_{0}=\vect{\Sigma}_{h}$. The parameter $\vect{\mu}_{h}$ is generated by letting driers play a full information repeated network congestion game that is stopped before converging to an equilibrium. 


\begin{figure}[h]
	\centering
	\includegraphics[clip,width=7cm]{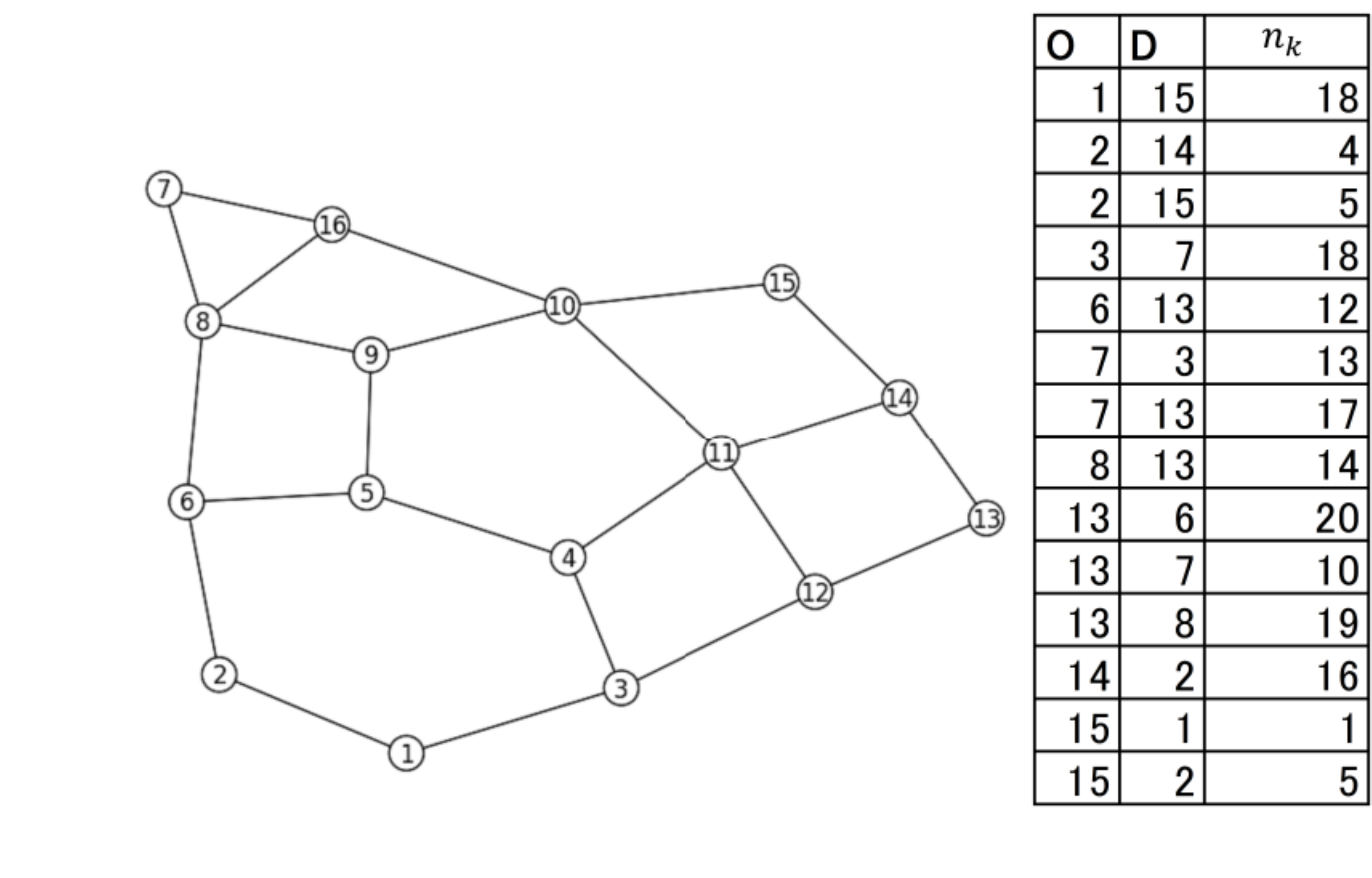}
	\capup
	\caption{Example network}
	\label{fig:nets}	
\end{figure}

In a manner similar to that in section~\ref{sec:eukp}, we change the signal parameter $\vect{\Sigma}_{s}$ and observe how traffic changes through the iterations. The initial $\vect{\mu}_{s}$ is randomly chosen and is updated by (\ref{eq:smap}). Since variance-covariance matrices are not full rank in this example, we use the dimension reduction technique in Appendix \ref{sec:gaussbayes} for calculations of inverse matrices.

Figure \ref{fig:unstablenet} shows the change of normalized traffic $\vect{f}^{t}/\|\vect{f}^{t}\|_{1}$ with unstable signal $\vect{\Sigma}_{s1}$ that is generated in the same way as $\vect{\Sigma}_{h}$. Since $\vect{f}^{t}$ is multi-dimensional, the figure shows the first component of the principal component analysis. The cumulative proportion of the first component is $18.3\%$. In this case, $\max |\lambda_{j}| = 2.2 > 1$, which does not satisfy the condition (\ref{eq:sfconcrp}) and then causes oscillation. Meanwhile, Figure \ref{fig:stablenet} shows the result with an asymptotically stable signal $\vect{\Sigma}_{s2}=9*\vect{\Sigma}_{s1}$. In this case, $\max |\lambda_{j}| = 0.65 < 1$ and the result quickly converges into a self-fulfilling traffic.

\begin{figure}[h]
	\centering
	\includegraphics[clip,width=7cm]{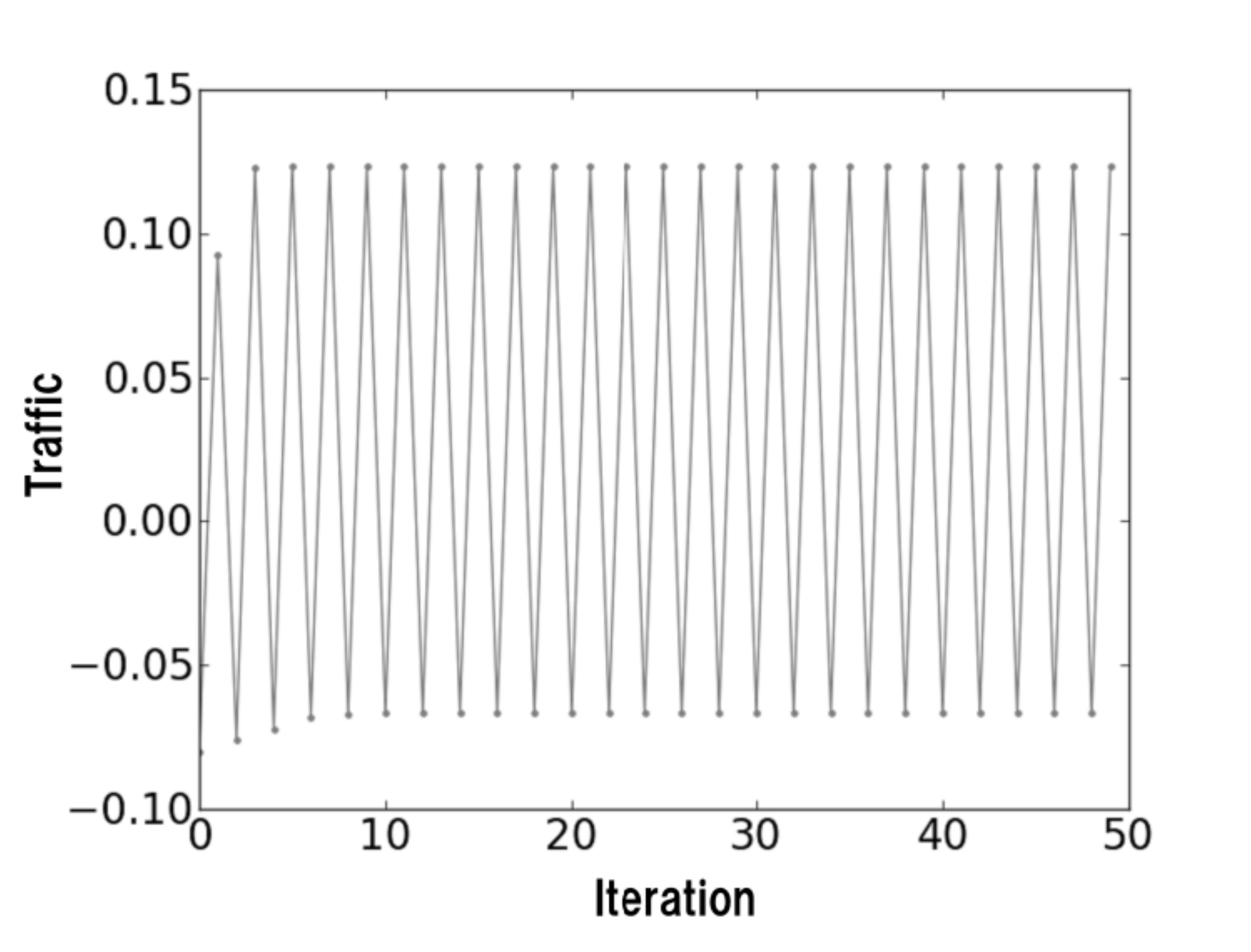}
	\capup
	\caption{Network oscillation caused by unstable signal}
	\label{fig:unstablenet}	
\end{figure}

\begin{figure}[h]
	\centering
	\includegraphics[clip,width=7cm]{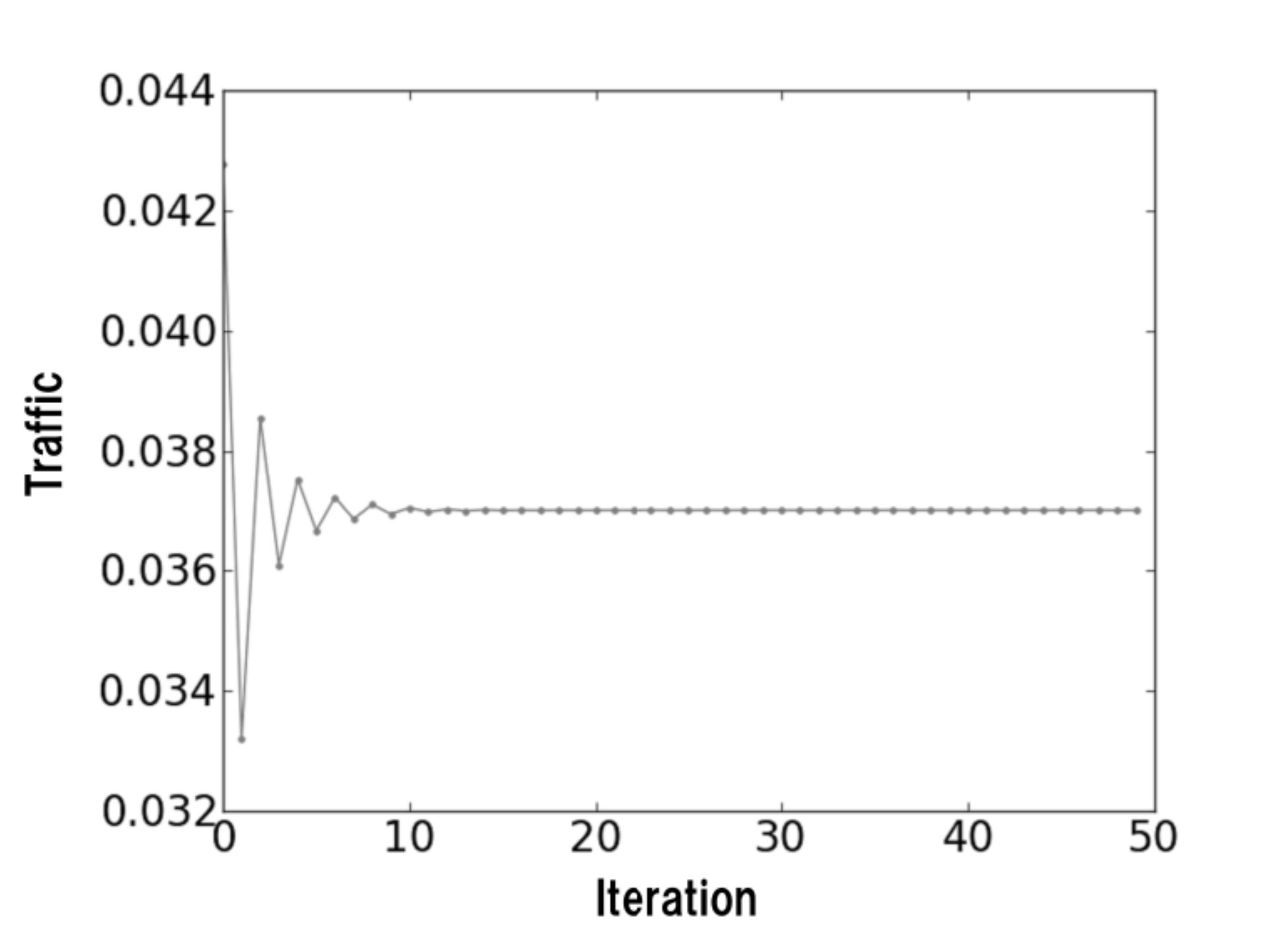}
	\capup
	\caption{Convergence of network traffic led by a self-fulfilling signal}
	\label{fig:stablenet}	
\end{figure}

\begin{figure}[!]
	\centering
	\includegraphics[clip,width=7cm]{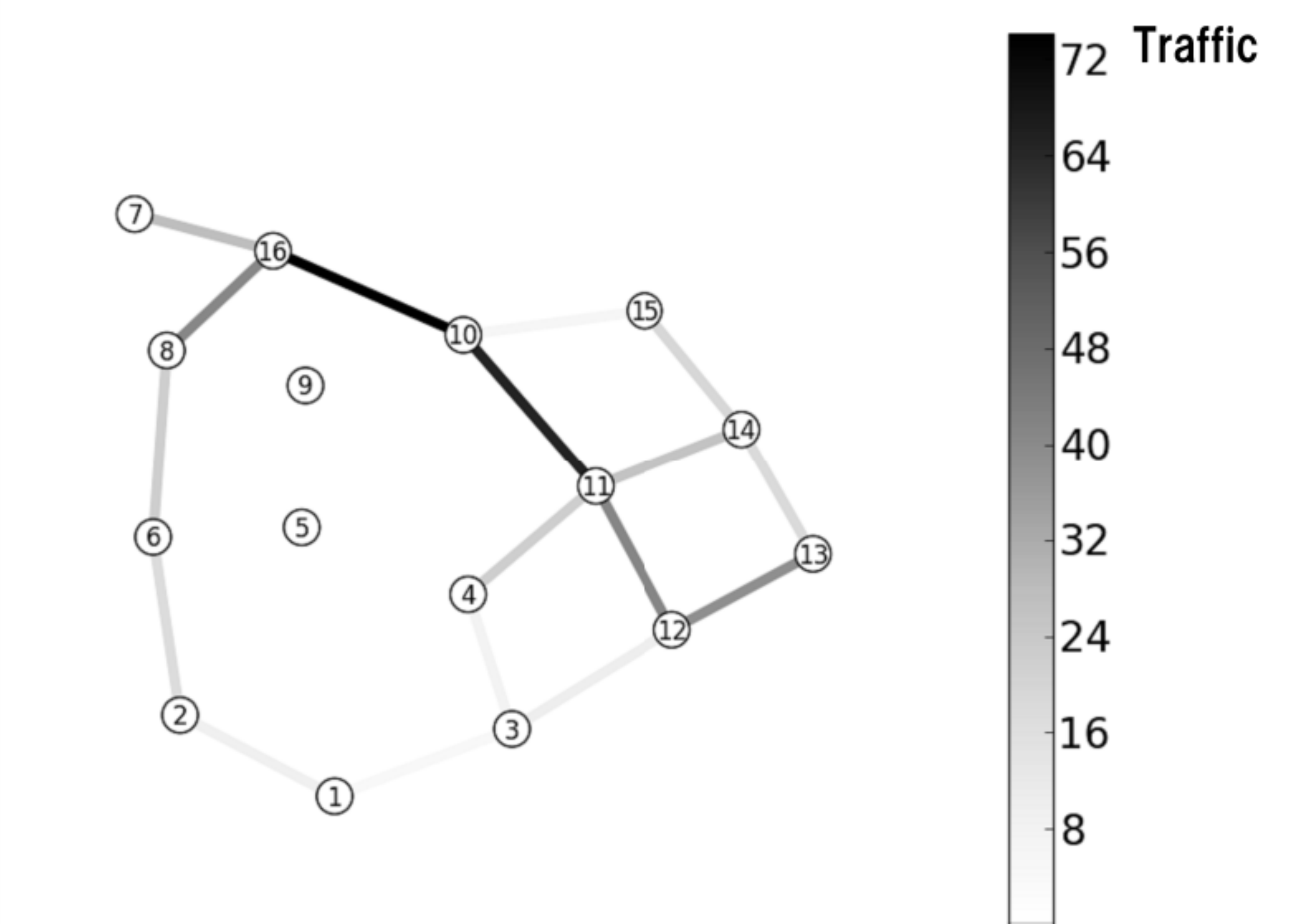}
	\capup
	\caption{Network traffic without signal}
	\label{fig:flow0}	
\end{figure}

\begin{figure}[!]
	\centering
	\includegraphics[clip,width=7cm]{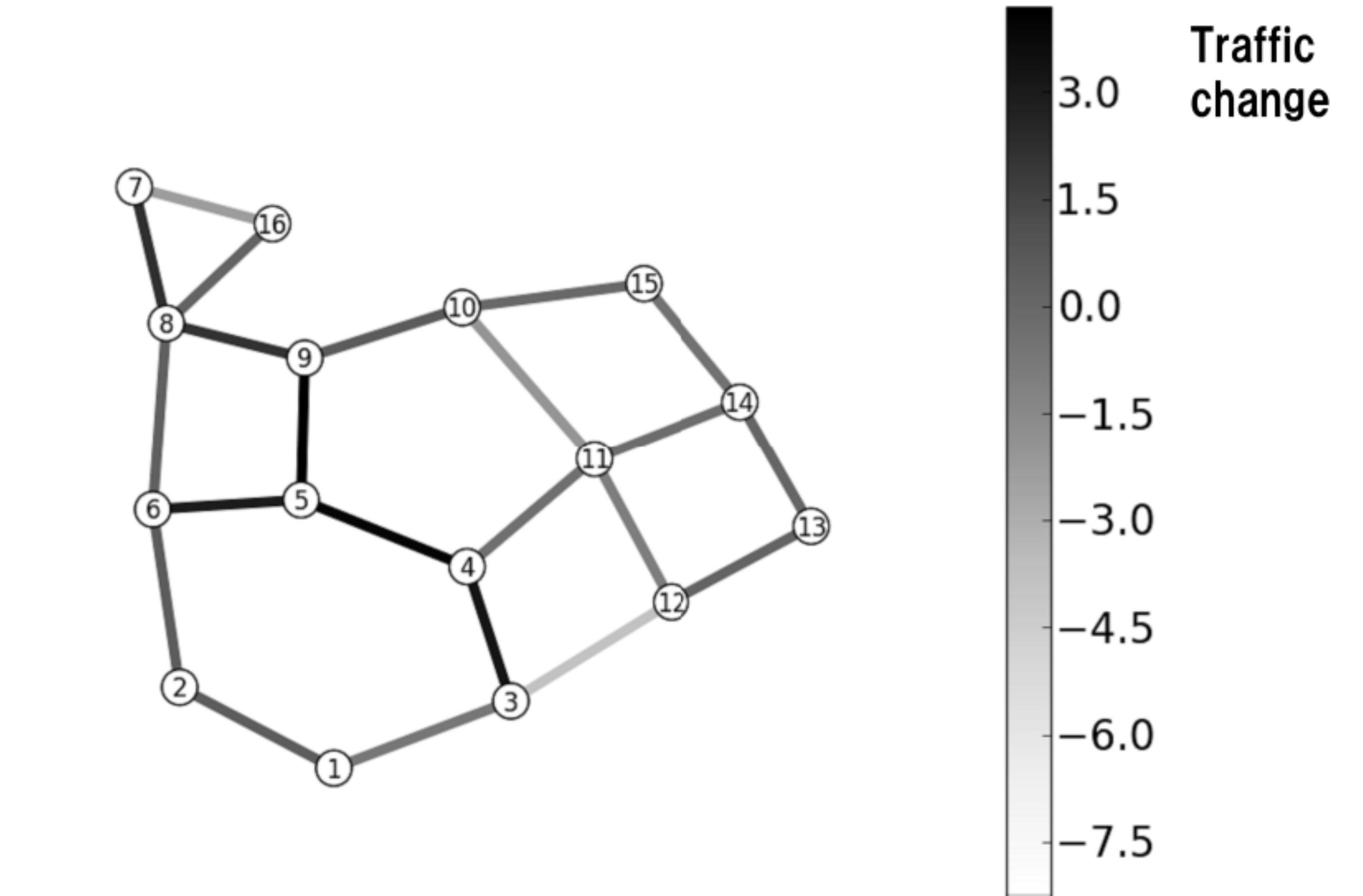}
	\caption{Change in network traffic caused by a self-fulfilling signal}
	\label{fig:flowd}	
\end{figure}

\begin{figure}[!]
	\centering
	\includegraphics[clip,width=7cm]{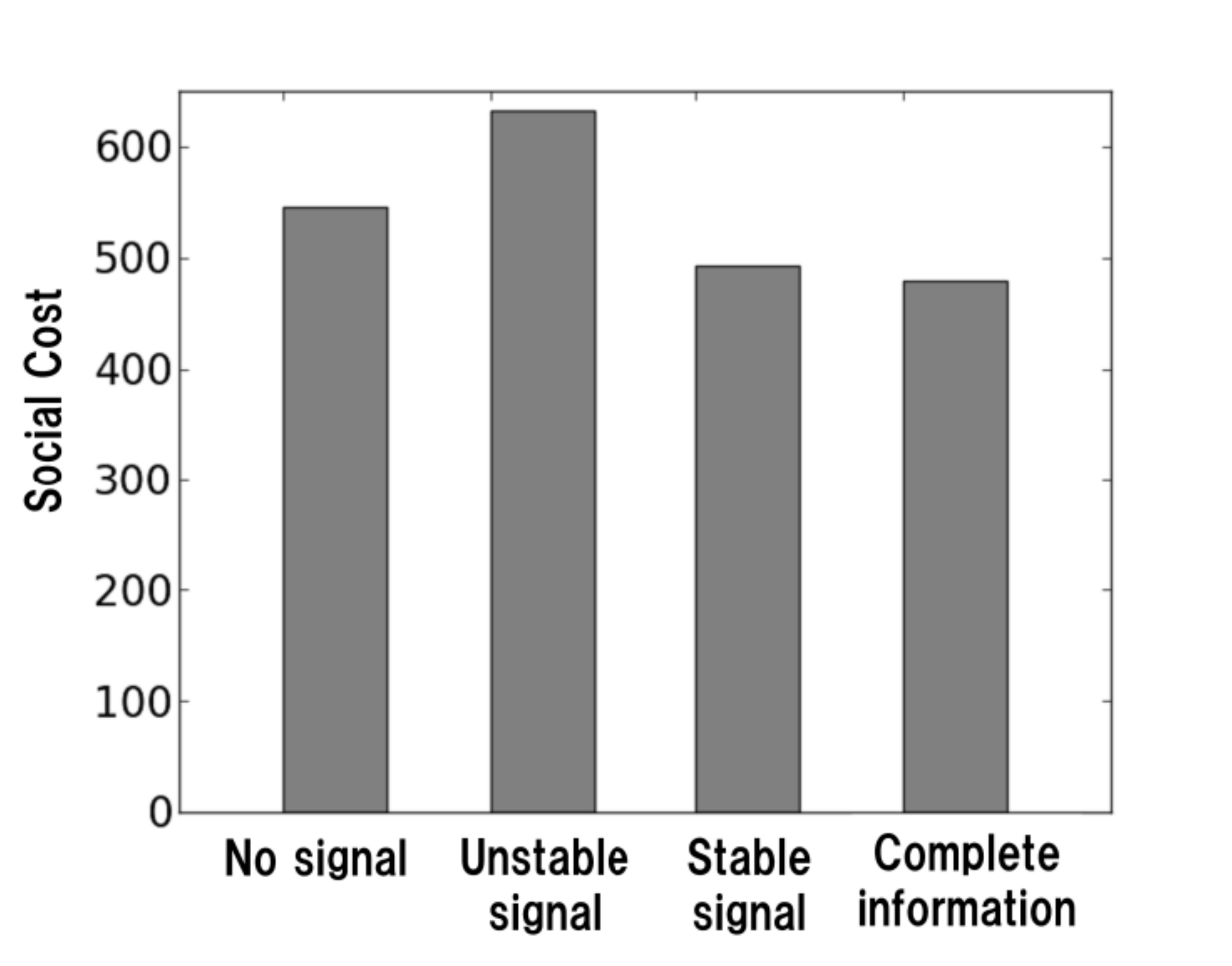}
	\capup
	\caption{Change in network traffic efficiency by a self-fulfilling signal}
	\label{fig:netsc}	
\end{figure}

Figure \ref{fig:flow0} shows network traffic without signal. It seems that most of the traffic is concentrated on certain edges. Figure \ref{fig:flowd} shows how self-fulfilling signals mitigate this problem. It seems that this signal induced traffic from congested edges to unused edges. Figure \ref{fig:netsc} indicates a comparison of social cost $\vect{\phi}\vect{h}$. Stable self-fulfilling signal successfully reduced social cost compared with the case without signal, whereas unstable signal increased the social cost. The rightmost bar in Figure \ref{fig:netsc} shows the social cost of Wardrop's UE in a complete information case, wherein drivers have full observability of network traffic without delay. The result of the stable self-fulfilling signal is the closest to the one of the complete information case.

\color{black}

\section{Conclusion}
\label{sec:conclusion}

In this study, we have shown how to coordinate agents using endogenous signal. This study formulates the problem of designing a self-fulfilling signal that suppresses the outcome oscillation and illustrates an instance of such a signal in a network congestion game. Examples demonstrate that the proposed technique is effective in coordination of agents and suppression of the oscillation.

Our theory is based on several assumptions which would be open problems for the future application to real-world problems. Specifically, following are interesting extensions:
\begin{itemize}
	\item Relaxation of the Gaussian prior assumption of receivers' beliefs
	\item Consideration of the receivers' reinforcement learning through repeated games
\end{itemize}



\bibliographystyle{named}
\bibliography{nets}

\newpage

\begin{appendices}
\section{Bayesian Inference for linear Gaussian model}
\label{sec:gaussbayes}
Here, we summarize the general results of Bayesian inference for a linear Gaussian model \citep{prml}. Let prior Gaussian distribution of $\vect{x}$ be denoted by
\begin{equation}
p(\vect{x})=N(\vect{x};\vect{\mu},\vect{\Lambda}^{-1}).
\end{equation}
In linear Gaussian model, the conditional distribution of $\vect{y}$ given $\vect{x}$ has a mean that is a linear function of $x$, such as
\begin{equation}
p(\vect{y}|\vect{x})=N(\vect{y};\vect{A}\vect{x}+\vect{b},\vect{L}^{-1}).
\end{equation}

The posterior distribution of $\vect{x}$ given $\vect{y}$ is denoted by
\begin{equation}
p(\vect{x}|\vect{y})=N(\vect{x};\vect{\Sigma}\{\vect{A}^T\vect{L}(\vect{y}-\vect{b})+\vect{\Lambda}\vect{\mu}\},\vect{\Sigma}),
\end{equation}
where
\begin{equation}
\vect{\Sigma}=(\vect{\Lambda}+\vect{A}^T\vect{L}\vect{A})^{-1}.
\end{equation}

The marginal distribution of $y$ is denoted by
\begin{equation}
p(\vect{y})=N(\vect{y};\vect{A}\vect{\mu}+\vect{b},\vect{L}^{-1}+\vect{A}\vect{\Lambda}^{-1}\vect{A}^T).
\end{equation}

\color{\cng}
 If covariance matrices are not full rank, it is impossible to calculate their inverse matrices. In this case, there exists a linear transformation that reduces the dimension of vector $\vect{x}$ and makes its covariance $\vect{\Sigma}$ full rank as follows
 
 	\begin{equation}
 	\left.
 	\begin{array}{l}
	\tilde{\vect{x}}=\vect{C}(\vect{x}-\vect{\mu}) \\
 	\tilde{\vect{\Sigma}}=\vect{C}\vect{\Sigma}\vect{C}^T.
 	\end{array}
 	\right.
 	\end{equation}
 	
This transformation matrix $\vect{C}$ is obtained by singular value decomposition. The calculation including inverse matrices can be processed in this reduced vector space. Let $\tilde{\vect{x}}_{2}$ be the result. This is converted into the original vector space by the following reverse transformation

 	\begin{equation}
 	\vect{x}_{2}=\vect{C}^T\tilde{\vect{x}}_{2}+\vect{\mu}.
 	\end{equation} 
  
\color{black}

\end{appendices}

\end{document}